\documentclass[a4,10pt,twocolumn]{scrartcl}

\pdfoutput=1

\usepackage[top=2.5cm, bottom=2.5cm, left=1.6cm, right=1.6cm]{geometry}
\pdfoutput=1

\input{myStandardPreamble_arXiv.tex}
\usepackage[english]{babel}
\usepackage[utf8x]{inputenc}
\usepackage{enumitem}					% mehr Möglichkeiten für Listen

\usepackage{libertine}
\usepackage[T1]{fontenc}
\usepackage[noBBpl]{mathpazo}			% option 'noBBpl' for double-lined \mathbb symbols 

\usepackage[hang]{footmisc}
\setfnsymbol{wiley}
\usepackage[T1]{fontenc}        % internal LaTeX font management, important for letters like ä,ö etc.             % langauge related package

\usepackage{graphicx}           % to include graphics (pdflatex
\usepackage{amsmath,amssymb}    % math packages, e.g. for align and pmatrix environments
\usepackage{color}
\usepackage{xcolor}   % you might consider using the dvipsnames option for more colors
\usepackage{microtype} %typographic enhancements
\usepackage{siunitx}
\usepackage{pbox}

\usepackage{marginnote}         % e.g. for url on the side
\usepackage{cite}               % improved citation management for numerical citations,

\usepackage{mathtools}
\makeatletter
\renewcommand{\@biblabel}[1]{[#1]\hfill}
\makeatother

\makeatletter
\let\NAT@parse\undefined
\makeatother

\usepackage[format=plain,			% no indentation in second line of caption
labelsep=period,		% period after ``Figure''
font=small,			% set font size
labelfont=bf,			% bold labels
skip=5pt			% spacing between caption and figure
]{caption}

\usepackage{authblk}

\newcommand{\newmarkedtheorem}[1]{%
	\newenvironment{#1}
	{\pushQED{\oprocend}\csname inner@#1\endcsname}
	{\popQED\csname endinner@#1\endcsname}%
	\newtheorem{inner@#1}%
}
\newtheorem{theorem}{Theorem}
\newtheorem{lemma}{Lemma}

\newtheorem{assumption}{Assumption}
\newtheorem{definition}{Definition}
\newtheorem{remark}{Remark}

\makeatletter
\def\th@plain{%
	\thm@notefont{}% same as heading font
	\normalfont % body font
}
\def\th@definition{%
	\thm@notefont{}% same as heading font
	\normalfont % body font
}
\makeatother

\addtokomafont{paragraph}{\itshape}

\setkomafont{section}{\Large}
\setkomafont{subsection}{\large}
\setkomafont{subsubsection}{\itshape}

\usepackage{tikz}				
\usetikzlibrary{arrows}				
\usetikzlibrary{shapes.misc}
\usetikzlibrary{patterns}
\usepackage{pgfplots}
\usetikzlibrary{calc}
\colorlet{istblue}{blue} % use structure theme to change
\colorlet{istred}{red!90!black}
\colorlet{istorange}{orange}
\colorlet{istgreen}{green!50!black}

\newcommand{\X}{\mathbb{X}}
\newcommand{\U}{\mathbb{U}}
\newcommand{\Z}{\mathbb{Z}}
\newcommand{\I}{\mathbb{I}}
\newcommand{\R}{\mathbb{R}}
\newcommand{\setS}{\mathbb{S}}

\newcommand{\pushright}[1]{\ifmeasuring@#1\else\omit\hfill$\displaystyle#1$\fi\ignorespaces}

\begin{document}

\title{Scheduling and control over networks using MPC with time-varying terminal ingredients: extended version$^*$} 
% Title, preferably not more than 10 words.

\date{}
\author{Stefan Wildhagen} 
\author{Frank Allg\"ower}
\affil{$^*$This paper is an extended version of \cite{Wildhagen20}. The authors would like to thank the German Research Foundation (DFG) for financial support within grant AL316/13-2 and within the German Excellence Strategy under grant EXC-2075. The authors are with the Institute for Systems Theory and Automatic Control, University of Stuttgart, Germany.
		{\small \{wildhagen}{\small,allgower\}@ist.uni}{\small -stuttgart.de}.}

\maketitle

\begin{abstract}        % Abstract of not more than 250 words.
\textbf{Abstract.} Rollout approaches are an effective tool to address the problem of co-designing the transmission schedule and the corresponding input values, when the controller is connected to the plant via a resource-constrained communication network. These approaches typically employ an MPC, activated at multiples of the period length of a base transmission schedule. Using multi-step invariant terminal regions and a prediction horizon longer than the base period, stability can be ensured. This strategy, however, suffers from intrinsic shortcomings, such as a high computational complexity and low robustness. We aim to resolve these drawbacks by proposing an MPC with periodically time-varying terminal region and cost for the rollout setup, where the controller is activated at each time instant and features an arbitrary but fixed prediction horizon. We consider in more detail two specific setups from the literature on Networked Control Systems, namely the token bucket network and actuator scheduling. For both setups, we provide conditions for which convergence under application of the time-varying MPC can be guaranteed. In two numerical examples, we demonstrate the benefits of the proposed method.
\end{abstract}

%===============================================================================

\section{Introduction}

Recently, there has been a great trend away from dedicated communication links between control applications, towards the use of shared and possibly wireless networks. Such communication networks outperform the traditional communication systems in terms of flexibility, low cost and ease of setting up. However, these advantages come with new challenges especially for control applications, since the communication links may in many cases not be considered ideal anymore. Instead, unfavorable effects such as a limited bandwidth, transmission delays or packet dropouts must be considered in the controller design process, giving rise to the research field of Networked Control Systems (NCSs).

An approach to counteract these limitations is to schedule access of the involved control applications to the network. The basic idea behind this approach is that when a transmission of sensor data or a new control input is only triggered if there are sufficient communication resources, delays and packet dropouts are much less likely to occur \cite{Walsh01}. An effective way to do this is via so-called \textit{rollout} approaches \cite{Antunes12_2,Antunes14,Peters16,Gommans17,Rosenthal18,Koegel19,Wildhagen19_2}, which are essentially a special case of model predictive control (MPC) schemes. Therein, a finite-horizon cost functional is minimized over the transmission schedule and the control inputs, before the computed optimal strategy is applied and the horizon is shifted forward in time. In case the resource-constrained network lies between controller and actuator, theoretical guarantees under such schemes were provided in \cite{Antunes14,Gommans17,Wildhagen19_2} for linear unconstrained, linear constrained and nonlinear constrained plants, respectively. Common to these approaches is that they use a \textit{multi-step} MPC formulation: assuming that the network may provide a periodic base schedule with a transmission every $M\ge 1$ time steps, the optimization problem is solved every $M$ steps and the optimal schedule and corresponding control inputs are applied in open loop in the meantime. Since transmitting a control every $M$ steps is feasible, an $M$-step control invariant terminal region and a terminal cost can be found, which are added to the optimization problem to ensure recursive feasibility and stability.

Such multi-step MPC formulations unfortunately suffer from some intrinsic drawbacks. Naturally, a prediction horizon of $N\ge M$ must be used due to the open-loop period of $M$ which is prescribed by the network and is not alterable. While already in general, long prediction horizons result in a high computational complexity, they especially do so for scheduling problems, since the integer transmission decisions induce an exponential growth in complexity with increasing horizon. Possible mitigations of this problem are i) to consider only \textit{a portion of} the possible transmission schedules in the optimization as proposed in \cite{Koegel19}, such that stability can still be guaranteed although optimality is lost; or ii) since the optimization is parallelizable, to run the computation on dedicated hardware, e.g., on field-programmable gate arrays (FPGAs). However, despite these measures, the basic problem remains. Furthermore, the long open-loop phase may result in a low robustness with respect to disturbances.

A traditional MPC scheme, solving an optimization problem with horizon $N\ge 1$ at every time step, would be advantageous especially in rollout control due to the aforementioned limitations of multi-step MPC. However, the conventional $1$-step control invariant terminal regions cannot be found in general in the rollout setup, since the feasible schedule only transmits a new control input every $M$ steps. A possible remedy is to employ an MPC solved at every time step with \textit{periodically time-varying terminal ingredients} as proposed previously in \cite{Boehm09,Lazar15,Lazar18}. Such a control scheme may overcome the mentioned shortcomings of multi-step MPC and at the same time guarantee recursive feasibility and stability for rollout approaches. The methods to find such $M$-periodic terminal regions and costs via convex programs provided in the mentioned references are, however, not directly applicable to rollout control due to the integer scheduling variable involved.

In this paper, we revisit two NCS setups already considered in the literature, in particular, that of control over a token bucket network \cite{Wildhagen19_2} and that of actuator scheduling \cite{Antunes12_2}. For the former, stability using a multi-step MPC has already been established while for the latter, no stability results have been obtained so far. Our main contribution is that for both NCS setups, we provide linear matrix inequality (LMI) conditions to compute $M$-periodic terminal regions and costs with which recursive feasibility and convergence under application of an MPC with time-varying terminal ingredients can be guaranteed. Furthermore, since the token bucket NCS is conveniently analyzed using methods of economic MPC, we give an extension of the results in \cite{Lazar18} to the general economic MPC setup. Stability in addition to convergence could be established for both NCS setups. However, a formal verification is cumbersome especially for the token bucket NCS, such that we focus on the notion of convergence throughout the paper.

The remainder of this paper is organized as follows. The time-varying terminal ingredients MPC scheme is introduced in Section \ref{sec_MPC}, where also conditions for convergence of the general setup are given. In Section \ref{sec_NCS}, the two considered NCS setups are briefly introduced and conditions for convergence under application of the time-varying MPC are established. We illustrate the theoretical results by two numerical examples in Section \ref{sec_num_results} before we conclude the paper in Section \ref{sec_summary}.

Let $\I$ and $\R$ denote the set of all integers and real numbers, respectively. We denote $\I_{[a,b]}\coloneqq\I\cap[a,b]$ and $\I_{\ge a}\coloneqq\I\cap[a,\infty)$, $a,b\in\I$, and $\mathbb{R}_{\ge a}\coloneqq[a,\infty)$, $a\in\R$. A function $\alpha:\mathbb{R}_{\ge 0}\rightarrow\mathbb{R}_{\ge 0}$ is said to be of class $\mathcal{K}_\infty$ if it is continuous, zero at zero, strictly increasing and unbounded. We denote by $I_n$ and $0_n$ the identity and zero matrix of dimension $n$, respectively, and by $0_{n\times m}$ the zero matrix of dimension $n$ by $m$. Whenever the dimension is clear from context, we omit these subscripts. Denote $A>0$ $(A\ge 0)$ a symmetric positive (semi-)definite matrix $A\in\R^{n\times n}$. For a vector $v\in\mathbb{R}^n$, the set distance to a subset $S\subseteq\mathbb{R}^n$ is defined as $||v||_S \coloneqq \min_{w\in S} ||v-w||$. For a function $f:\R^n\rightarrow\R^n$, define the image of the subset $S$ as $f(S)\coloneqq\{f(x):x\in S\}$.

\section{MPC with periodically time-varying terminal ingredients}
\label{sec_MPC}

\subsection{MPC Setup}

Consider a nonlinear discrete-time system
\begin{equation}
	x(k+1) = f(x(k),u(k)), \label{system_gen}
\end{equation}
with the state $x(k)\in\X\subseteq\R^n$ and input $u(k)\in\U\subseteq\R^m$ at time $k\in\I_{\ge 0}$, where $f:\X\times\U\rightarrow\X$ is continuous. The state and input constraint sets are assumed to be closed.

The objective function of the MPC at time $k$ is defined by
\begin{equation*}
	V(x(\cdot | k),u(\cdot | k),k)= \sum_{i=0}^{N-1} \ell(x(i|k),u(i|k)) + V_f(x(N|k),k),
\end{equation*}
where $x(\cdot | k)\coloneqq\{x(0|k),\ldots,x(N|k)\}$ and $u(\cdot | k)\coloneqq\{u(0|k),\ldots,u(N-1|k)\}$ denote the predicted state and input trajectories, respectively. The stage cost $\ell: \X\times\U\rightarrow\R$ and terminal cost $V_f:\X\times\I_{\ge 0}\rightarrow\R$ are both assumed to be continuous. The MPC operates with the following scheme:
\begin{enumerate}
	\item At time $k$, measure $x(k)$, solve the optimization problem denoted by $\mathcal{P}(x(k),k)$
	\begin{align}
		V^*(&x(k),k) =\min_{x(\cdot|k),u(\cdot|k)} V(x(\cdot | k),u(\cdot | k),k) \nonumber \\
		\text{s.t. }	&x(i+1|k)=f(x(i|k),u(i|k)) \nonumber \\
		&x(i|k)\in\X,\; u(i|k)\in\U, \quad \forall i\in\I_{[0,N-1]} \nonumber \\
		&x(0|k)=x(k),\; x(N|k)\in\X_f(k) \nonumber
	\end{align}
	and denote its minimizers by $x^*(\cdot|k)$ and $u^*(\cdot|k)$.
	\item Apply the feedback $u(k)=u^*(0|k)$ to \eqref{system_gen}.
	\item Set $k\leftarrow k+1$ and go to 1).
\end{enumerate}
Note that both the terminal cost $V_f(\cdot,k)$ and the closed terminal region $\X_f(k)\subseteq\X, \;k\in\I_{\ge 0}$ are allowed to be time-varying. We denote by $\mathcal{X}(k)$ the feasible set at time $k$, i.e., the set of all states $x$ such that $\mathcal{P}(x,k)$ admits a solution.

Typically in MPC, the goal is to design the terminal region and cost such that recursive feasibility and stability can be guaranteed. As introduced in \cite{Lazar18}, a periodic terminal ingredients design is employed here. For an $M\in\I_{\ge 1}$, let $\{\setS_j,F_j(\cdot)\}_{j=0,\ldots,M-1}$ denote a set of $M$ terminal ingredients, where $\setS_j$ are terminal regions and $F_j:\setS_j\rightarrow\R$ are terminal costs. Suppose now that $\X_f(0)=\setS_j$ and $V_f(\cdot,0) = F_j(\cdot)$, for an arbitrary $j\in\I_{[0,M-1]}$ are used in $\mathcal{P}(x(0),0)$ at initial time. Then, the time-varying terminal cost and region at all other time instances are defined by
\begin{equation*}
	\X_f(k) \coloneqq \setS_{(j+k)\text{mod}M} \text{ and } V_f(\cdot,k) \coloneqq F_{(j+k)\text{mod}M}(\cdot).
\end{equation*}

\subsection{Recursive Feasibility and Convergence}

In some scenarios arising in NCSs, the stage cost may not be positive definite, such that the usual conditions ensuring convergence and stability in MPC are not fulfilled. In our recent work \cite{Wildhagen19_2}, it was demonstrated that such a scenario may arise in particular if the communication network involves a dynamical component itself. Recently, there has been a great interest in the analysis of \textit{economic} MPC (see, e.g., \cite{Amrit11,Dong18}), where a general (non-positive-definite) stage cost is used, such that the optimal regime of operation of the system might lie in a general subset $\bar{\X}$ of the state space. Conditions have been developed in the former references under which convergence to this optimal regime of operation can be guaranteed. Due to its great usefulness in the analysis of NCSs, an extension to economic MPC of the setup presented in \cite{Lazar18} is given next. Since many of the introduced concepts are well-known in the literature on economic MPC, we will keep the proofs short and refer the interested reader to related literature instead.

Of central importance in economic MPC are the concepts of optimal asymptotic behavior of a system and dissipativity, which we will introduce next.
\begin{definition}
	The lowest possible asymptotic average cost $\ell_{av}^*$ is defined by
	\begin{equation*}
		\ell^*_{av} \coloneqq \inf_{x(0)\in\X} \hspace{-2pt} \inf_{\substack{u(\cdot) \\ x(k+1)=f(x(k),u(k)) \\ (x(k),u(k))\in\X\times\U}} \hspace{-2pt} \liminf_{K\rightarrow\infty}  \frac{\sum_{k=0}^{K-1} \ell(x(k),u(k))}{K}.
	\end{equation*}
\end{definition}
\vspace{2pt}
Hence, $\ell_{av}^*$ denotes the lowest attainable asymptotic cost of the system under its dynamics and constraints.
\begin{assumption}
	System \eqref{system_gen} is strictly dissipative with respect to the control invariant set $\bar{\X}$ and the supply rate $\ell(x,u)-\ell^*_{av}$, i.e., there exists a continuous storage function $\lambda: \X\rightarrow\mathbb{R}_{\ge 0}$ and a $\mathcal{K}_\infty$-function $\rho$ such that for all $(x,u)\in\X\times\U$
	\begin{equation*}
		\ell(x,u)-\ell_{av}^*+\lambda(x)-\lambda(f(x,u)) \ge \rho(||x||_{\bar{\X}}).
	\end{equation*}
	\label{ass_dissi}
\end{assumption}
\vspace{-10pt}
To ensure recursive feasibility, a notion of periodic terminal invariance is required due to the periodically time-varying terminal regions. For the approach of multi-step MPC, in contrast, an $M$-step invariant terminal region is required, see e.g. \cite{Koegel13,Lazar15}.
\begin{assumption}[\cite{Lazar18}]
	There exist an $M\in\I_{\ge 1}$, terminal regions $\setS_j, \; j=0,\ldots,M-1$ and terminal controllers $\kappa_j:\setS_j\rightarrow\U, \; j=0,\ldots,M-1$ such that
	\begin{itemize}
		\item $f(\setS_j,\kappa_j(\setS_j))\subseteq\setS_{j+1}, \; \forall j\in\I_{[0,M-2]}$,
		\item $f(\setS_{M-1},\kappa_{M-1}(\setS_{M-1}))\subseteq\setS_{0}$,
		\item $\bar{\X}\subseteq\setS_j\subseteq\X, \; \forall j\in\I_{[0,M-1]}$.
	\end{itemize}
	\label{ass_term_inv}
\end{assumption}
For convergence, we also require a decrease condition on the terminal costs and a further assumption on the relation of terminal cost and storage function.
\begin{assumption}
	There exist terminal costs $F_j:\setS_j\rightarrow\R, \; j=0,\ldots,M-1$ such that with $M$, $\setS_j$ and  $\kappa_j, \; j=0,\ldots,M-1$ from Assumption \ref{ass_term_inv}
	\begin{itemize}
		\item $F_{j+1}(f(x,\kappa_j(x)))-F_j(x)\le-\ell(x,\kappa_j(x)) + \ell_{av}^*, \; \forall x\in\setS_j, \; \forall j\in\I_{[0,M-2]}$,
		\item $F_{0}(f(x,\kappa_{M-1}(x)))-F_{M-1}(x)\le-\ell(x,\kappa_{M-1}(x)) + \ell_{av}^*, \; \forall x\in\setS_{M-1}$.
	\end{itemize}
	\label{ass_term_decr}
\end{assumption}
Note that for $M=1$, Assumptions \ref{ass_term_inv} and \ref{ass_term_decr} recover the standard assumptions used in MPC, cf. \cite{Amrit11,Dong18}.
\begin{assumption}
	The minimum of $F_j(x)+\lambda(x)$ is $0$ for all $j\in \I_{[0,M-1]}$. The minimums are attained exactly on $\bar{\X}$, i.e., $\bar{\X}=\arg\min_{x}F_j(x)+\lambda(x)$ for all $j\in \I_{[0,M-1]}$.
	\label{ass_lb_term_cost_rot}
\end{assumption}
\begin{remark}
	The minimums at $0$ are without loss of generality.
\end{remark}

With these assumptions, we are ready to establish recursive feasibility and convergence to $\bar{\X}$.
\begin{theorem}
	Suppose $x(0)\in\mathcal{X}(0)$. If Assumptions \ref{ass_dissi}-\ref{ass_lb_term_cost_rot} hold, then the optimization problem $\mathcal{P}(x(k),k)$ is feasible for all $k\in\I_{\ge 0}$ and $x(k)$ converges to $\bar{\X}$ as $k\rightarrow\infty$.
	\label{thm_convergence_MPC}
\end{theorem}
\begin{proof}
	Note first that for each $k\in\I_{\ge 0}$, there exists a $p\in\I_{[0,M-1]}$ such that $(j+k)\text{mod}M=p$, i.e., $\X_f(k)=\setS_p$ and $V_f(\cdot,k)=F_p(\cdot)$ are used as terminal ingredients in $\mathcal{P}(x(k),k)$.
	
	It is then possible to establish recursive feasibility as in \cite{Lazar15,Lazar18}, by considering the feasible input trajectory at $k+1$
	\begin{equation}
		u(\cdot | k + 1)=\{u^*(1|k),\ldots,u^*(N-1|k),\kappa_p(x^*(N|k))\}. \label{feas_inp_traj}
	\end{equation}
	
	For the proof of convergence, we rely on the notion of rotated cost functions as usually done in economic MPC (see, e.g., \cite{Amrit11}). The rotated optimization problem can be shown to have the same minimizer as $\mathcal{P}$, such that the former can be used for analysis of convergence. Considering the feasible input trajectory \eqref{feas_inp_traj}, the proof of convergence is then essentially a direct combination of the proofs of \cite[Theorem 3]{Lazar18} and \cite[Theorem 15]{Amrit11}, and is omitted here for brevity.
\end{proof}
\begin{remark}
	Under some additional assumptions, also stability of $\bar{\X}$ can be shown. We do not elaborate on stability here and refer the reader to \cite[Section 2.4.5]{Rawlings17}.
\end{remark}

\section{Application to Networked Control Systems}
\label{sec_NCS}

In this section, we revisit two setups that were already considered in the literature on NCSs. For both setups, we investigate in detail how Assumptions \ref{ass_dissi}-\ref{ass_lb_term_cost_rot}, which guarantee convergence of the time-varying MPC scheme, can be fulfilled. Thereby, we also give guidelines to design suitable terminal ingredients.

\subsection{Control over a dynamical token bucket network}

\begin{figure}
	\centering
	\begin{tikzpicture}[>= latex]
	%\tikzmath{\xmax = 7; \ymax =-2; \slength = 0.25;}
	\node [thick,draw,rectangle, inner sep=1.5pt,minimum size=1.1cm,align = center,rounded corners=3pt] (actuator) at (0.2*7,0) {ZOH Actuator};
	\node[anchor=south east,inner sep=1pt] at (actuator.south east) {\textcolor{istblue}{$u_s$}};
	\node [thick,draw,rectangle, inner sep=1.5pt,minimum size=1.1cm,align = center,rounded corners=3pt] (plant) at (0.55*7,0) {Plant};
	\node[anchor=south east,inner sep=1pt] at (plant.south east) {\textcolor{istblue}{$x_p$}};
	\node [thick,draw,rectangle, inner sep=1.5pt,minimum size = 1.1cm,align = center,rounded corners=3pt] (ctrl) at (0.83*7,0) {MPC};
	\node [thick,draw,rectangle, inner sep=1.5pt,minimum size = 1.1cm,align = center,rounded corners=3pt] (netw) at (0.5*7,-1.7) {Token Bucket Network};
	\node[anchor=south east,inner sep=1pt] at (netw.south east) {\textcolor{istblue}{$\beta$}};
	
	\node (upper) at (7,0.2*-2){};
	\node (lower) at (7,0.2*-2-0.25){};
	\node (lower2) at (7+0.25*0.5,0.2*-2 -0.25*0.866){};
	
	\path (0,0) --++ (7,0) node(helpne){} --++ (0,-1.7) node(helpse){} --++ (-7,0) node(helpsw){} --++ (0,1.7) node(helpnw){};
	
	\draw [->,thick] (plant.east) -- (ctrl.west) node[pos = 0.5,above] {$x_p$};
	\draw [thick] (ctrl.east) -- (helpne.center)node[pos=1.2,above] {$u_c$};
	\draw [thick] (helpne.center) -- (upper.center) -- (lower2.center) node[pos=0.5,right] {$\gamma$};
	\draw[thick] (lower.center) -- (helpse.center) -- (netw.east);
	\draw[->,thick] (netw.west) -- (helpsw.center) -- (helpnw.center) -- (actuator.west);
	\draw[->,thick] (actuator.east) -- (plant.west) node[pos = 0.5,above] {$u_s$};
	\end{tikzpicture}
	\caption{Considered configuration of the NCS with token bucket network.}
	\label{fig_network_tb}
\end{figure}
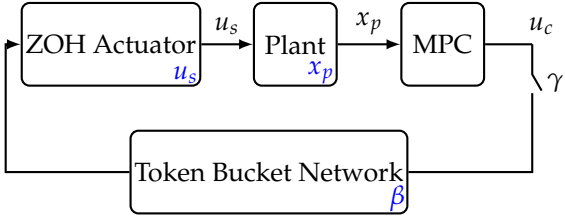
An effective scheduling and control co-design can greatly improve control performance over a resource-constrained network. The token bucket network is a simple and common model for such a network in information theory. A depiction of the NCS configuration considered here can be found in Figure \ref{fig_network_tb}. The state of the controlled plant, $x_p$, is sensed at the controller, which determines both the control input $u_c$ and the transmission decision $\gamma$. The control inputs are transmitted over a network to a zero-order hold (ZOH) actuator, which may only hold the last transmitted input and store it in a state $u_s$. The network itself is described by the so-called token bucket specification, where a transmission over the network is only possible if the associated bucket level state $\beta$ is high enough to support the cost of a transmission. The decision $\gamma$ indicates whether a new control input was sent over the network, where $\gamma=1$ if a transmission was triggered, and $\gamma=0$ if not.

Denoting the overall state and control by
\begin{equation*}
	x\coloneqq[x_p^\top \; u_s^\top \;  \beta ]^\top \text{ and } u\coloneqq[ u_c^\top \; \gamma ]^\top,
\end{equation*}
we consider a dynamical model of the NCS setup which readily incorporates the communication constraints and takes the form
\begin{equation}
	f(x,u) \coloneqq \begin{bmatrix} Ax_p+ B(1-\gamma)u_s + B\gamma u_c \\
		(1-\gamma) u_s + \gamma u_c\\
		\min\{\beta+g-\gamma c,b\}
	\end{bmatrix}.
	\label{dynamics_tb}
\end{equation}
The first line in \eqref{dynamics_tb} corresponds to the dynamics of the linear plant with matrices $A\in\R^{n_p\times n_p}$ and $B\in\R^{n_p\times m_p}$, at which either a new transmitted input is applied immediately or the saved input is held. The second line describes the dynamics of the ZOH actuator, while the third line corresponds to the dynamics of the bucket level. New tokens are added to the current bucket level at a constant rate of $g\in\I_{\ge 1}$ and the cost of a transmission is $c\in\I_{\ge g}$.	If the maximum bucket capacity $b\in\I_{\ge c}$ is reached, arriving tokens are discarded.

The state and input of the overall NCS is constrained by $x(k)\in\X\coloneqq\X_p\times\U_p\times\I_{[0,b]}\subseteq\R^n$ and $u(k)\in\U\coloneqq\U_p\times\{0,1\}\subseteq\R^m$, where $n=n_p+m_p+1$ and $m=m_p+1$. Here, $\X_p$ and $\U_p$ denote the closed plant state and input constraint sets, where $\X_p\times\U_p$ contains the origin. The constraint $\beta(k)\in\I_{[0,b]}$ is to ensure that a transmission is only triggered if the bucket level is high enough to support the cost of a transmission.

The stage cost associated with the NCS is given by
\begin{equation*}
	\ell(x,u) \coloneqq x_p^\top Q x_p + (1-\gamma)u_s^\top R u_s + \gamma u_c^\top R u_c, \;Q,R>0.
\end{equation*}
Note that due to the transmission decision $\gamma$, $\ell$ is not positive definite with respect to the state $u_s$. The bucket level is not considered in the cost at all. Such a cost arises when only the state and input of the plant are part of the performance measure. This is a natural choice since typically, performance of the plant should be maximized under consideration of the communication constraints. We refer to \cite{Wildhagen19_2} for a more detailed description of this NCS setup.
\begin{remark}
	To apply the MPC scheme, the sensor only needs to measure the plant state $x_p(k)$. For known initial conditions of $u_s$ and $\beta$, the controller can simply keep track of these states by simulation.
\end{remark}

A crucial characteristic of the token bucket network is that if at initial time $\beta(0)\in\I_{[c-g,b]}$, it can be guaranteed that a schedule transmitting every $\lceil\frac{c}{g}\rceil$ time instances is feasible. Inspired by this observation, we choose the $M\coloneqq \lceil\frac{c}{g}\rceil$ terminal regions to
\begin{equation}
	\setS_0 = \underbrace{\{0\}\times\{0\}\times\I_{[0,c-g-1]}}_{\eqqcolon\setS_{0}'}\cup\underbrace{\Z_0\times\I_{[c-g,b]}}_{\eqqcolon\setS_{0}''},
	\label{S0_tb}
\end{equation}
and for all $j\in\I_{[1,M-1]}$ to
\begin{equation}
	\setS_j = \{0\}\times\{0\}\times\I_{[0,(j-1)g-1]} \cup \Z_j \times \I_{[(j-1)g,b]},
	\label{Sj_tb}
\end{equation}
with $\I_{[0,-1]}\coloneqq\emptyset$ and sets $\Z_j\subseteq\X_p\times\U_p, \; j=0,\ldots,M-1$. With these terminal regions, consider the terminal controller
\begin{equation}
	\kappa_0(x)=\begin{cases}
		[ 0 \; 0 ]^\top & x\in\setS_{0}' \\
		[ (K[\substack{x_p \\ u_s}])^\top \; 1 ]^\top & x\in\setS_{0}''
	\end{cases}
	\label{term_control_tb}
\end{equation}
with $K\in\R^{m_p\times n_p+m_p}$, and $\kappa_j(x)=[0\; 0]^\top, \; \forall j\in\I_{[1,M-1]}$. At $j=0$, a control is transmitted if $x\in\setS_{0}''$, whereas if $x\in\setS_{0}'$, i.e., if the plant state and saved input are already in the equilibrium, no transmission is triggered. At all other instances in the period, no transmission is triggered as well.

We define
\begin{equation*}
	\tilde{A} \coloneqq \begin{bmatrix} A & 0 \\
		0 & 0 \end{bmatrix}, \;
	\tilde{B} \coloneqq \begin{bmatrix} B \\
		I \end{bmatrix} \text{ and }
	A' \coloneqq \begin{bmatrix} A & B \\
		0 & I \end{bmatrix}.
\end{equation*}
Furthermore, we define $A'' = \tilde{A}+\tilde{B}K$ and note that $A''$ describes the dynamics of the plant state and saved input under the control $\kappa_0$ if $x\in\setS_{0}''$, and that $A'$ does so under the controls $\kappa_j, \; \forall j\in\I_{[1,M-1]}$. With these definitions and an invariance assumption on $\Z_j$, the following lemma establishes periodic invariance of the chosen terminal regions.
\begin{assumption}
	The sets $\Z_j, j\in\I_{[0,M-1]}$ are closed and each contain the origin. Furthermore, they satisfy
	$A'' \Z_0\subseteq\Z_1$, $A' \Z_j\subseteq\Z_{j+1}, \; j\in\I_{[1,M-2]}$ and $A' \Z_{M-1}\subseteq\Z_0$.
	\label{ass_term_inv_tb}
\end{assumption}
\begin{lemma}
	Suppose that Assumption \ref{ass_term_inv_tb} holds. Then,
	\begin{itemize}
		\item Assumption \ref{ass_dissi} is fulfilled with $\ell_{av}^*=0$, $\lambda(x)=||u_s||_S^2$, $R\ge S>0$ and $\bar{\X}\coloneqq \{0\}\times\{0\}\times\I_{[0,b]}$,
		\item Assumption \ref{ass_term_inv} is fulfilled with the terminal regions $\setS_j$ and the terminal controllers $\kappa_j$ as defined in \eqref{S0_tb},\eqref{Sj_tb},\eqref{term_control_tb}.
	\end{itemize}
	\label{lemma_term_invariance_tb}
\end{lemma}
\begin{proof}
	For the first statement, we refer the reader to the proof of \cite[Theorem 1]{Wildhagen19_2}.
	
	For the second statement, consider first the case $j=0$. With the terminal controller \eqref{term_control_tb}, we have $f(\setS_{0}',\kappa_0(\setS_{0}'))=\{0\}\times\{0\}\times\I_{[0,c-1]}$ and
	$f(\setS_{0}'',\kappa_0(\setS_{0}''))=
	A'' \Z_0 \times \min\{\I_{[c-g,b]}+g-c,b\} \subseteq \Z_1 \times\I_{[0,b]}$ by $A''\Z_0\subseteq\Z_1$,
	such that in summary, $f(\setS_0,\kappa_0(\setS_{0})) = f(\setS_{0}',\kappa_0(\setS_{0}')) \cup f(\setS_{0}'',\kappa_0(\setS_{0}'')) \subseteq \setS_1$.
	
	For $j\in\I_{[1,M-2]}$, we establish $f(\setS_j,\kappa_j(\setS_j)) = f(\setS_j,[0\; 0]^\top)
	= \{0\}\times\{0\} \times \min\{\I_{[0,(j-1)g-1]}+g,b\} \cup A' \Z_j \times \min\{\I_{[(j-1)g,b]}+g,b\} \subseteq \setS_{j+1}$ by $A'\Z_j\subseteq\Z_{j+1}$.
	
	For $j=M-1$, we have with $A'\Z_{M-1}\subseteq\Z_{0}$ $f(\setS_{M-1},\kappa_{M-1}(\setS_{M-1})) = \{0\}\times\{0\} \times \min\{\I_{[0,(M-2)g-1]}+g,b\} \cup A' \Z_{M-1} \times \min\{\I_{[(M-2)g,b]}+g,b\} \subseteq  \{0\}\times\{0\}\times\I_{[0,(M-1)g-1]} \cup \Z_0 \times \I_{[(M-1)g,b]} \subseteq\setS_0$. The last estimate holds since $M=\lceil\frac{c}{g}\rceil\ge\frac{c}{g}$.
	
	The terminal controllers fulfill the input constraints directly since $A'' \Z_0\subseteq\Z_1\subseteq\X_p\times\U_p$ due to Assumption \ref{ass_term_inv_tb}.
	
	Lastly, $\bar{\X}\subseteq\setS_j\subseteq\X, \; \forall j\in\I_{[0,M-1]}$ holds since all $\Z_j$ contain the origin and $\Z_j\subseteq\X_p\times\U_p, \; \forall j\in\I_{[0,M-1]}$.
\end{proof}

For the terminal costs, consider
\begin{equation}
	F_j(x) = [\substack{x_p \\ u_s}]^\top P_j [\substack{x_p \\ u_s}], \; \forall j\in\I_{[0,M-1]}
	\label{term_cost_tb}
\end{equation}
with $P_j\in\R^{n_p+m_p\times n_p+m_p}$. The following result provides a choice for the matrices $P_j$ such that convergence of the NCS with token bucket network under application of the MPC scheme can be established.
\begin{theorem}
	Suppose that $x(0)\in\mathcal{X}(0)$, that Assumption \ref{ass_term_inv_tb} holds and that there exist symmetric positive definite $X_j\in\R^{n_p+m_p\times n_p+m_p}$, and $Y\in\R^{m_p\times n_p+m_p}$ such that the LMIs
	\begin{align}
		&{\small \begin{bmatrix}
				X_1 & 0 & 0 & \tilde{A}X_0 + \tilde{B}Y \\
				0 & Q^{-1} & 0 & X_0[\substack{I \\ 0}] \\
				0 & 0 & R^{-1} & Y \\
				X_0\tilde{A}^\top + Y^\top\tilde{B}^\top & [\substack{I \, 0}]X_0 & Y^\top & X_0
		\end{bmatrix} } \ge 0, \label{LMI_tb_1}\\
		&{\small \begin{bmatrix}
				X_{(j+1)\text{mod}M} & 0 & A' X_j \\
				0 & \text{diag}(Q^{-1},R^{-1}) & X_j \\
				X_j A'^\top & X_j & X_j \\
		\end{bmatrix}} \ge 0 \label{LMI_tb_2}
	\end{align}
	are satisfied for all $j\in\I_{[1,M-1]}$. Then with $P_j = X_j^{-1}$ and terminal regions and costs as defined in \eqref{S0_tb},\eqref{Sj_tb} and \eqref{term_cost_tb}, the MPC optimization problem $\mathcal{P}(x(k),k)$ is feasible for all $k\in\I_{\ge 0}$ and $x_p(k)$ and $u_s(k)$ converge to $0$ as $k\rightarrow\infty$.
	\label{thm_convergence_tb}
\end{theorem}
\begin{proof}
	Let the control gain be defined by $K=Y X_0^{-1}$. We apply the Schur complement to \eqref{LMI_tb_1} and \eqref{LMI_tb_2} and substitute $X_j=P_j^{-1}$ and $Y=K P_0^{-1}$. Pre- and postmultiplying $P_j$ and inserting $A'' = \tilde{A}+\tilde{B}K$ then gives
	\begin{equation}
		\begin{aligned}
			&A''^\top P_1 A'' - P_0 + [\substack{Q \; 0 \\ 0 \; 0}] + K^\top[\substack{R \; 0 \\ 0 \; 0}] K \le 0, \\
			&A'^\top P_{(j+1)\text{mod}M} A' - P_j + [\substack{Q \; 0 \\ 0 \; R}] \le 0, \; j\in\I_{[1,M-1]}.
		\end{aligned} 	\label{LMI_tb_3}
	\end{equation}
	Since due to Lemma \ref{lemma_term_invariance_tb}, $\ell_{av}^*=0$, these conditions imply that Assumption \ref{ass_term_decr} is fulfilled with this particular choice of terminal costs. Furthermore, since $P_j>0$ for all $j\in\I_{[0,M-1]}$ and $\lambda(x)=||u_s||_S^2$, $R\ge S>0$ from Lemma \ref{lemma_term_invariance_tb}, we conclude that Assumption \ref{ass_lb_term_cost_rot} is fulfilled as well.
	
	With Lemma \ref{lemma_term_invariance_tb}, all conditions of Theorem \ref{thm_convergence_MPC} are fulfilled, thus $\mathcal{P}(x(k),k)$ is feasible for all $k\in\I_{\ge 0}$ and the closed loop state converges to $\bar{\X}=\{0\}\times\{0\}\times\I_{[0,b]}$ as $k\rightarrow\infty$.
\end{proof}
\begin{remark}
	Suppose that the plant state and input constraint sets are polytopic, i.e., $\X_p \coloneqq \{x_p\in\R^{n_p}: c_i x_p \le 1, i=1,\ldots,p_x\}$ and $\U_p \coloneqq \{u_s\in\R^{m_p}: d_i u_s \le 1, i=1,\ldots,p_u\}$ with $c_i\in\R^{1\times n_p}$ and $d_i\in\R^{1\times m_p}$. Then, two possible ways to construct sets $\Z_j$ that fulfill Assumption \ref{ass_term_inv_tb} are the following:
	First, consider polytopic $\Z_j$.
	\begin{enumerate}
		\item Determine an invariant polytope $\Z\subseteq\X_p\times\U_p$ such that $A'^{M-1}A''\Z\subseteq\Z$, e.g., by \cite[Algorithm 2]{Pluymers05}.
		\item Determine $\alpha^* = \max_{\alpha\in[0,1]} \alpha \text{ s.t. } \alpha A'^{j-1}A''\Z\subseteq\X_p\times\U_p$ for all $j\in\I_{[1,M-1]}$.
		\item Define $\Z_0\coloneqq\alpha^*\Z$ and $\Z_j\coloneqq \alpha^* A'^{j-1}A''\Z$ for all $j\in\I_{[1,M-1]}$.
	\end{enumerate}
	Second, consider ellipsoidal $\Z_j\coloneqq\{z\in\R^{n_p+m_p}: z^\top P_j z\le \alpha\}, \; \forall j\in\I_{[0,M-1]}$ with $\alpha\in\R_{\ge 0}$.
	\begin{enumerate}
		\item In view of the system dynamics of $x_p$ and $u_s$ under the terminal controller, note that \eqref{LMI_tb_3} implies that if \\ \vspace{3pt} $\big[\substack{x_p(j) \\ u_s(j)}\big]\in \Z_j$, then also $\big[\substack{x_p(j+1) \\ u_s(j+1)}\big]\in \Z_{(j+1)\text{mod}M}$. \vspace{2pt}
		\item To ensure that $\Z_j\subseteq\X_p\times\U_p$, we add
		\begin{align}
			&{\small\begin{bmatrix}
					1 & [c_{i_x} \; 0] X_j \\
					X_j [c_{i_x} \; 0]^\top & \alpha^{-1}X_j
				\end{bmatrix} \ge 0,} \label{LMI_tb_ellipse} \\
			&{\small\begin{bmatrix}
					1 & [0 \; d_{i_u}] X_j \\
					X_j [0 \; d_{i_u}]^\top & \alpha^{-1}X_j
				\end{bmatrix} \ge 0,
				\begin{bmatrix}
					1 & d_{i_u} Y \\
					Y^\top d_{i_u}^\top  & \alpha^{-1}X_0
				\end{bmatrix} \ge 0} \nonumber
		\end{align}
		for all $j\in\I_{[0,M-1]}$, $i_x\in\I_{[0,p_x]}$ and $i_u\in\I_{[0,p_u]}$ to the collection of LMIs \eqref{LMI_tb_1} and \eqref{LMI_tb_2}. To maximize the volume of the terminal sets, we maximize $\alpha$. Techniques on how to modify \eqref{LMI_tb_1}, \eqref{LMI_tb_2} and \eqref{LMI_tb_ellipse} such that a simultaneous search for $X_j,Y$ and $\alpha$ is possible via a convex program can be found in \cite{Boehm09,Lazar18}.
	\end{enumerate}
	\label{rem_construction_term_regions}
\end{remark}
\begin{remark}
	For nonlinear plants, we expect that LMIs to search for terminal regions and costs could be found in a similar fashion as proposed in \cite{Lazar18}. To arrive at resilient statements, however, this issue warrants further research.
	\label{remark_nonlinear}
\end{remark}

\subsection{Actuator scheduling}

The actuator scheduling setup, which was first introduced in \cite{Antunes12_2}, is depicted in Figure \ref{fig_network_act}. Therein, a plant with multiple actuators is considered. A controller, which is collocated with the plant's sensors, measures the plant state $x_p$ and transmits new control inputs to the actuators over a network that only one user may access at the same time. As a result, the controller can only send new control inputs to one actuator at a given sampling time. Hence, the controller must compute both a control input $u_c$ and a decision $\sigma$, which indicates which actuator should receive a new control input. Assuming there are $M$ actuators, the scheduling decision may take any value in $\sigma\in\I_{[0,M-1]}$. As in the original setup, we assume a set-to-zero strategy, i.e., the input is set to zero if the corresponding actuator is not scheduled.

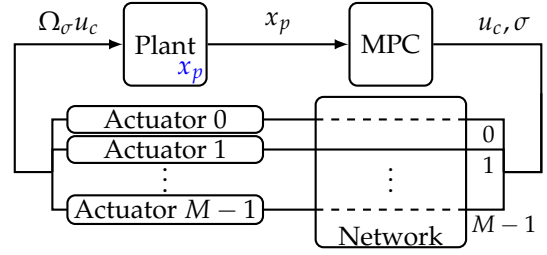
\begin{figure}
	\centering
	\begin{tikzpicture}[>= latex]
	%\tikzmath{\xmax = 7; \ymax =-2; \slength = 0.25;}
	\node [thick,draw,rectangle, inner sep=1.5pt,minimum height=0.3cm,minimum width = 2.6cm,align = center,rounded corners=3pt] (actuator) at (2,-1.0) {Actuator 0};
	\node [thick,draw,rectangle, inner sep=1.5pt,minimum height=0.3cm,minimum width = 2.6cm,align = center,rounded corners=3pt] (actuator) at (2,-1.4) {Actuator 1};
	\node [thick,draw,rectangle, inner sep=1.5pt,minimum height=0.3cm,minimum width = 2.6cm,align = center,rounded corners=3pt] (actuator) at (2,-2.2) {Actuator $M-1$};
	
	\node [thick,draw,rectangle, inner sep=1.5pt,minimum size=1.1cm,align = center,rounded corners=3pt] (plant) at (2,0) {Plant};
	\node[anchor=south east,inner sep=1pt] at (plant.south east) {\textcolor{istblue}{$x_p$}};
	\node [thick,draw,rectangle, inner sep=1.5pt,minimum size = 1.1cm,align = center,rounded corners=3pt] (ctrl) at (5,0) {MPC};
	\node [thick,draw,rectangle, inner sep=1.5pt,minimum height = 2cm, minimum width = 2cm,align = center,rounded corners=3pt] (netw) at (5,-1.7) {};
	\node[anchor=south,inner sep=1pt] at (netw.south) {Network};
	
	\node (upper) at (7,0.2*-2){};
	\node (lower) at (7,0.2*-2-0.25){};
	\node (lower2) at (7+0.25*0.5,0.2*-2 -0.25*0.866){};
	
	\path (0,0) --++ (7,0) node(helpne){} --++ (0,-1.7) node(helpse){} --++ (-7,0) node(helpsw){} --++ (0,1.7) node(helpnw){};
	
	\draw [->,thick] (plant.east) -- (ctrl.west) node[pos = 0.5,above] {$x_p$};
	\draw [thick] (ctrl.east) -- (helpne.center) -- (7,-1.7) -- (6.5,-1.7);
	\draw[thick] (lower.center) -- (helpse.center) -- (6.5,-1.7);
	\draw[->,thick] (0.5,-1.7) -- (0,-1.7) -- (0,0) -- (plant.west) node[pos = 0.5,above] {$\Omega_\sigma u_c$};
	
	\draw[thick] (6.5,-1.7) -- (6.5,-1.0) -- (6,-1.0);
	\draw[thick] (6.5,-1.7) -- (6.5,-1.4) -- (6,-1.4);
	\draw[thick] (6.5,-1.7) -- (6.5,-2.2) -- (6,-2.2);
	
	\draw[thick,dashed] (6,-1.0) -- (4,-1.0);
	\draw[thick] (6,-1.4) -- (4,-1.4);
	\draw[thick,dashed] (6,-2.2) -- (4,-2.2);
	\node at (5,-1.7){$\vdots$};
	\node at (2,-1.7){$\vdots$};
	
	\draw[thick] (4,-1.0) -- (3.3,-1.0);
	\draw[thick] (4,-1.4) -- (3.3,-1.4);
	\draw[thick] (4,-2.2) -- (3.3,-2.2);
	
	\node at (6.3,-1.2){\small $0$};
	\node at (6.3,-1.6){\small$1$};
	\node at (6.5,-2.4){\small$M-1$};
	\node at (6.5,0.25){$u_c,\sigma$};
	
	\draw[thick] (0.7,-1.0) -- (0.5,-1.0) -- (0.5,-1.7);
	\draw[thick] (0.7,-1.4) -- (0.5,-1.4) -- (0.5,-1.7);
	\draw[thick] (0.7,-2.2) -- (0.5,-2.2) -- (0.5,-1.7);
	\end{tikzpicture}
	\caption{Considered configuration of the NCS with actuator scheduling.}
	\label{fig_network_act}
\end{figure}
With the overall state and control given by
\begin{equation*}
	x \coloneqq x_p \text{ and } u \coloneqq [u_c^\top \sigma]^\top,
\end{equation*}
the NCS can be described by a dynamical model of the form
\begin{equation*}
	f(x,u) = Ax_p + B\Omega_\sigma u_c
\end{equation*}
with $A\in\R^{n_p\times n_p}$ and $B\in\R^{n_p\times m_p}$. The control input $u_c$ is partitioned as $u_c=[u_c^{0,\top},\ldots,u_c^{M-1,\top}]^\top$, where $u_c^j\in\R^{m_{p,j}}$ is affiliated with the $j$-th actuator. The matrix $\Omega_\sigma\in\R^{m_p\times m_p}$ is defined by
\begin{equation*}
	\Omega_\sigma \coloneqq \text{diag}(0_{\sum_{i=0}^{\sigma-1}m_{p,i}},I_{m_{p,\sigma}},0_{\sum_{i=\sigma+1}^{M-1}m_{p,i}}),
\end{equation*}
i.e., it selects which components of the input may act on the plant. We do not consider constraints on the plant state and control input as in \cite{Antunes12_2}, such that $\X \coloneqq \R^{n_p}$, $\U \coloneqq \R^{m_p}\times\I_{[0,M-1]}$. The cost associated with the plant is
\begin{equation*}
	\ell(x,u) = x_p^\top Q x_p + u_c^\top\Omega_\sigma R u_c,
\end{equation*}
where $Q>0$ and $R=\text{diag}(R_0,\ldots,R_{M-1})$, $0<R_j\in\R^{m_{p,j}\times m_{p,j}}$ for all $j\in\I_{[0,M-1]}$. The setup is adapted from \cite{Antunes12_2}, where in addition, output measurements and process and measurement noise were considered. No stability results were obtained in this paper, where the focus was on methods to solve the involved optimization problem explicitly. However, to guarantee convergence, we focus on the case of an ideal process model and state measurement here.

Suppose that there exists a base schedule
\begin{equation*}
	\Sigma_B = \{\sigma_0,\ldots,\sigma_{M-1}\}, \; \sigma_j \in\I_{[0,M-1]},\; \forall j\in\I_{[0,M-1]}.
\end{equation*}
With this base schedule, consider the terminal controllers
\begin{equation*}
	\kappa_j(x) = \begin{bmatrix} K_j x_p \\ \sigma_{j} \end{bmatrix} , \; \forall j \in\I_{[0,M-1]}
\end{equation*}
with $K_j\in\R^{m_p\times n_p}$. Consider the terminal regions and costs
\begin{equation}
	\setS_j=\R^{n_p} \text{ and } F_j(x) = x_p^\top P_j x_p, \; j\in\I_{[0,M-1]} \label{term_ing_act}
\end{equation}
with $P_j\in\R^{n_p\times n_p}$. With these terminal ingredients, the following result gives conditions such that the plant converges under the MPC scheme.
\begin{theorem}
	Suppose that there exist symmetric positive definite $X_j\in\R^{n_p\times n_p}$, and $Y_j\in\R^{m_p\times n_p}$ such that the LMIs
	\begin{equation}
		{\small \begin{bmatrix}
				X_{(j+1)\text{mod}M} & 0 & 0 & AX_j + B\Omega_{\sigma_{j}} Y_j \\
				0 & Q^{-1} & 0 & X_j \\
				0 & 0 & R_{\sigma_{j}}^{-1} & \Pi_{\sigma_{j}} Y_j \\
				X_j A^\top + Y_j^\top \Omega_{\sigma_{j}} B^\top & X_j & Y_j^\top\Pi_{\sigma_{j}}^\top & X_j
		\end{bmatrix}} \vspace{-1pt} \ge \vspace{-1pt} 0 \label{LMI_act}\\
	\end{equation}
	are satisfied for all $j\in\I_{[0,M-1]}$, where
	\begin{equation*}
		\Pi_{\sigma} \coloneqq \begin{bmatrix} 0_{m_{p,\sigma}\times\sum_{i=0}^{\sigma-1} m_{p,i}} & I_{m_{p,\sigma}} & 0_{m_{p,\sigma}\times\sum_{i=\sigma+1}^{M-1} m_{p,i}}
		\end{bmatrix}.
	\end{equation*}
	Then with $P_j=X_j^{-1}$ and the terminal regions and costs as defined in \eqref{term_ing_act}, $x_p(k)$ converges to $0$ as $k\rightarrow\infty$.
	\label{thm_convergence_act}
\end{theorem}
\begin{proof}
	First, note that Assumption \ref{ass_dissi} is trivially fulfilled with $\ell_{av}^*=0$, $\lambda(x)=0$ and $\bar{\X}=0$. Further, Assumption \ref{ass_term_inv} is trivially fulfilled as well.
	
	Let the control gains be defined by $K_j=Y_j X_j^{-1}$. We apply the Schur complement to \eqref{LMI_act}, substitute $X_j=P_j^{-1}$ and $Y_j=K_j P_j^{-1}$, and pre- and postmultiply $P_j$ to obtain
	\begin{align*}
		(A + B\Omega_{\sigma_j}K_j)&^\top P_{(j+1)\text{mod}M}(A + B\Omega_{\sigma_j}K_j) -P_j\\
		&+ Q + K_j^\top \Omega_{\sigma_j} R K_j \le 0, \; \forall j\in\I_{[0,M-1]}.
	\end{align*}
	We conclude that Assumption \ref{ass_term_decr} is fulfilled with this choice of terminal costs. Assumption \ref{ass_lb_term_cost_rot} is fulfilled since $P_j>0$.
	
	Again, all assumptions of Theorem \ref{thm_convergence_MPC} are fulfilled such that we conclude convergence of $x_p(k)$ to $0$.
\end{proof}
\begin{remark}
	Note that since there are neither plant state, control input nor terminal constraints, $\mathcal{P}$ is always feasible.
\end{remark}
\begin{remark}
	Owing to the original setup and for ease of presentation, no state and input constraints were assumed for the actuator scheduling setup. Indeed, polytopic $\X_p$ and $\U_p$ could be handled by introducing terminal regions that fulfill similar conditions as described in Remark \ref{rem_construction_term_regions}.
\end{remark}

\section{Numerical results}
\label{sec_num_results}

\subsection{Token bucket network}

For numerical analysis, we revisit the example in \cite{Wildhagen19_2}. Consider the linearized batch reactor taken from \cite{Walsh01}, discretized with a sampling period of \SI{0.1}{\second}, and assume that control inputs are transmitted over a token bucket network. The bucket has the parameters $g=1$, $c=8$ and $b=22$, which results in $M=8$. The cost matrices are chosen to $Q=10I_4$ and $R=I_2$, and the initial conditions are given by $x_p(0)=[1\; 0\; 1\; 0]^\top$, $u_s(0)=[0\; 0]^\top$ and $\beta(0)=22$. Other than in the numerical example in \cite{Wildhagen19_2}, we consider here
\begin{equation*}
\X_p = [-2,2]^4 \text{ and } \U_p = [-3,3]^2
\end{equation*}
as state and input constraints, which necessitate terminal constraints to ensure recursive feasibility and convergence.

To compute suitable terminal costs and controllers, the LMIs provided in Theorem \ref{thm_convergence_tb} were solved using Matlab, Yalmip \cite{YALMIP} and SDPT3 \cite{SDTP3}. We compare the approach of multi-step MPC as taken in \cite{Wildhagen19_2} with the time-varying terminal ingredients MPC presented in this paper. For these approaches, it is necessary to come up with an $M$-step invariant set or $M$ periodically invariant sets, respectively. For this numerical example, we computed these sets using the first approach presented in Remark \ref{rem_construction_term_regions} using Matlab and MPT3 \cite{MPT3}. Note that the set $\setS_0$ obtained from this algorithm is also an adequate $M$-step invariant set for the multi-step MPC. For comparison, we let the time-varying MPC start with $\setS_0$ as a terminal region at time $k=0$ as well. Note that for the multi-step MPC, a prediction horizon of $N\ge M=8$ is needed. This may already result in high computation times to solve the MPC optimization problem due to the (in the worst case) $2^N$ integer decisions involved. The purpose of this numerical example is to show that this limitation can be greatly attenuated by a time-varying terminal ingredients MPC. This scheme allows in principle for any prediction horizon $N\ge 1$, regardless of the inalterable $M$.

\begin{table}
	\caption{Elapsed time while solving MPC optimization problem for multi-step and time-varying terminal ingredients MPC.}
	\begin{tabular}{c | c c c}
		N & {\footnotesize Multi-step MPC} & {\footnotesize Time-varying MPC} & \pbox{2cm}{{\footnotesize Relative Comp. Time of time-varying MPC}} \\ \hline
		12 & \SI{2.954}{\second} & \SI{2.244}{\second} & \SI{437.4}{\percent} \\
		10 & \SI{1.357}{\second} & \SI{1.028}{\second} & \SI{200.4}{\percent} \\
		8 & \SI{0.648}{\second} & \SI{0.513}{\second} & \SI{100}{\percent} \\
		6 & - & \SI{0.158}{\second} & \SI{30.8}{\percent} \\
		4 & - & \SI{0.082}{\second} & \SI{16.0}{\percent} \\
		2 & - & \SI{0.024}{\second} & \SI{4.7}{\percent}
	\end{tabular}
	\label{tbl_tb}
\end{table}
Table \ref{tbl_tb} shows the elapsed time while solving the MPC optimization problem for the multi-step MPC and the time-varying terminal ingredients MPC, respectively, depending on the chosen prediction horizon $N$. For the multi-step MPC, we considered the optimization problem at initial time only, whereas for the time-varying MPC, an average over the first $M$ instances was taken. For comparability, no warm start was used in the optimization in the latter case. Firstly, one can see that the computation times for larger horizons $N=12,10,8$ are approximately the same for both approaches, and quite high in this example. The optimizations were run on an Intel Core i7 and not on dedicated hardware, but nonetheless for rapidly sampled systems, such high computation times may jeopardize practical implementation of the control and scheduling scheme. For the multi-step MPC, using even smaller prediction horizons is not possible due to the constraint $N\ge M=8$. For the time-varying terminal ingredients MPC proposed in this paper, shorter prediction horizons are indeed possible. Note from the last column of Table \ref{tbl_tb} that the elapsed time drastically reduces for shorter horizons, up to a factor of $20$ for $N=2$ as compared to $N=8$. However, shorter prediction horizons also shrink the feasible set. In this example, no solution to the optimization problem could be found for the considered initial conditions and $N=1$. In case that the plant is unconstrained, this limitation vanishes since also no terminal region is required then, such that the optimization problem is always feasible.

\subsection{Actuator scheduling}

For the actuator scheduling setup, we consider a plant that is comprised of two dynamically decoupled batch reactors from \cite{Walsh01}, which is again discretized with a period of \SI{0.1}{\second}. The initial conditions are identical for both reactors $x(0) = [1\; 0\; 1\; 0\; 1\; 0\; 1\; 0]^\top$ and we do not consider state and input constraints. There are $M=4$ control inputs to the system with $m_{p,j}=1$ each, such that $m_p=4$. We assume the cost matrices are given by
\begin{equation*}
	Q = \text{diag}(I_4,10 I_4) \text{ and } R = \text{diag}(10,0.1,1,1),
\end{equation*}
i.e., the state of the second reactor is weighted higher than that of the first. Further, the first input of the first reactor is penalized much more than its second input, whereas the inputs are weighted equally for the second reactor.

\begin{figure}
	\centering
	% This file was created by matlab2tikz.
%
%The latest updates can be retrieved from
%  http://www.mathworks.com/matlabcentral/fileexchange/22022-matlab2tikz-matlab2tikz
%where you can also make suggestions and rate matlab2tikz.
%
%
\begin{tikzpicture}

\begin{axis}[%
width=\columnwidth,
height=1.436in,
xmin=0,
xmax=30,
xlabel style={font=\color{white!15!black}},
ymin=-2,
ymax=2,
ylabel style={font=\color{white!15!black}},
axis background/.style={fill=white}
]
\addplot[const plot, color=istgreen, forget plot] table[row sep=crcr] {%
0	1\\
1	1.68976501003236\\
2	1.38433496255331\\
3	0.676140492655575\\
4	0.375690384383584\\
5	0.246063895583125\\
6	0.191148459318153\\
7	0.169685489529069\\
8	0.122189318513458\\
9	0.0645023694653307\\
10	0.0319296266262983\\
11	0.0109963058917136\\
12	-0.0043158570777787\\
13	-0.0175850015021914\\
14	-0.0169582671399418\\
15	-0.0108290738251241\\
16	-0.0114755716258896\\
17	-0.00820774601045215\\
18	-0.00355813274924421\\
19	-0.00304277250540763\\
20	-0.00424995398108918\\
21	-0.006239680902532\\
22	-0.00517991819175823\\
23	-0.00294119006118221\\
24	-0.00254463442781505\\
25	-0.001722065526929\\
26	-0.000737220348951739\\
27	-0.000476101686112118\\
28	-0.000514064018108051\\
29	-0.000384414592463633\\
};
\addplot[const plot, color=istorange, forget plot] table[row sep=crcr] {%
0	0\\
1	-0.0624391831297391\\
2	-0.117064386651565\\
3	-0.130992410495177\\
4	-0.12593788023278\\
5	0.00147712114206311\\
6	-0.0278929143816923\\
7	0.051057555049447\\
8	0.0113673103745472\\
9	0.0434092545804027\\
10	0.0568843336496073\\
11	0.0626249752682255\\
12	0.0341173997676043\\
13	0.0333442622258399\\
14	0.0186855862803134\\
15	0.00970982607102607\\
16	0.00430372598553834\\
17	0.00112115447374284\\
18	0.00490821773083061\\
19	0.00656468959053647\\
20	0.00696436833975123\\
21	0.00391285905613357\\
22	0.00213537813578446\\
23	0.00250689802149695\\
24	0.0023347608562343\\
25	0.00131805672843555\\
26	0.00121036367263216\\
27	0.00106856929800792\\
28	0.000932779287927787\\
29	0.000540051955458476\\
};
\addplot[const plot, color=istblue, forget plot] table[row sep=crcr] {%
0	1\\
1	0.636786853078522\\
2	-1.90945796167701\\
3	-1.02521365706411\\
4	-0.665470163631382\\
5	-0.468041922966507\\
6	-0.376513109051232\\
7	-0.291234601358185\\
8	-0.361565893101762\\
9	-0.2643951890381\\
10	-0.195451466374841\\
11	-0.141011649704119\\
12	-0.109160584892057\\
13	-0.0837417162489131\\
14	-0.0321373428275302\\
15	-0.0372833932088159\\
16	-0.0377580088309894\\
17	-0.0173544436913513\\
18	-0.0214787147519713\\
19	-0.0199649121610224\\
20	-0.0165521814986613\\
21	-0.0141192263903315\\
22	-0.00290323567150175\\
23	-0.00501174709742913\\
24	-0.00502996282845602\\
25	-0.00137206477097274\\
26	-0.00209741994467394\\
27	-0.00201481202437725\\
28	-0.00166880481146214\\
29	-0.000599548243750192\\
};
\addplot[const plot, color=istred, forget plot] table[row sep=crcr] {%
0	0\\
1	0.0886821265819751\\
2	-0.0548143932941074\\
3	-0.257162885305943\\
4	-0.354321970367756\\
5	-0.348784950931788\\
6	-0.33772031291971\\
7	-0.286222406971445\\
8	-0.260267794434767\\
9	-0.224940036756437\\
10	-0.180856120547742\\
11	-0.13581507769546\\
12	-0.106968741132355\\
13	-0.0819211842970937\\
14	-0.0633869930253297\\
15	-0.0504379352177583\\
16	-0.0429110648256438\\
17	-0.0370250914786168\\
18	-0.0299638691136135\\
19	-0.0235641015236636\\
20	-0.0179146676089923\\
21	-0.0143578919004737\\
22	-0.0114787830264871\\
23	-0.00860715205356177\\
24	-0.00646200655965168\\
25	-0.00492567526853449\\
26	-0.00361909108633918\\
27	-0.00265993810527165\\
28	-0.00192491707106712\\
29	-0.00141633728656099\\
};
\end{axis}

\begin{axis}[%
width=\columnwidth,
height=1.436in,
at={(0cm,-2.7cm)},
xmin=0,
xmax=30,
xlabel style={font=\color{white!15!black}},
xlabel={$k$},
ymin=-2,
ymax=2,
ylabel style={font=\color{white!15!black}},
axis background/.style={fill=white}
]
\addplot[const plot, color=istgreen, forget plot] table[row sep=crcr] {%
0	1\\
1	1.05645023133529\\
2	0.711319145289463\\
3	0.452409261685577\\
4	0.233673981450337\\
5	0.146536757682211\\
6	0.0867421277702697\\
7	0.0401292689151472\\
8	0.0212718500456778\\
9	0.0147027248934283\\
10	0.0138774764647203\\
11	0.0160643829510967\\
12	0.0127211290331052\\
13	0.00739316918409134\\
14	0.00626934782194533\\
15	0.0038193122982233\\
16	0.000858939862518171\\
17	-0.000380440813841525\\
18	-0.0010063061669691\\
19	-0.00143708896515959\\
20	-0.00184687044047161\\
21	-0.00142218515693204\\
22	-0.000668866979545004\\
23	-0.000355782972610707\\
24	-0.000230380161935134\\
25	-0.000186323789233122\\
26	-0.000179004493060767\\
27	-0.000189958956273374\\
28	-0.000211886076553325\\
29	-0.000242343180357836\\
};
\addplot[const plot, color=istorange, forget plot] table[row sep=crcr] {%
0	0\\
1	-0.0534412139015044\\
2	-0.0796798285041458\\
3	-0.0894263506251847\\
4	0.232782672857166\\
5	0.137410873443264\\
6	0.0811020123508599\\
7	0.0485596331426108\\
8	0.0296555311110565\\
9	0.018469257521872\\
10	0.0117034752695916\\
11	0.00747983606238354\\
12	0.00481224412820902\\
13	0.00332477017780645\\
14	0.00246728791835174\\
15	0.00194544361200899\\
16	-0.00594745382398032\\
17	-0.00363254080000959\\
18	-0.00219758000173485\\
19	-0.00130531460285837\\
20	-0.000743028872933415\\
21	-0.000390649535248588\\
22	-0.000198150336332562\\
23	-9.54064921596278e-05\\
24	-3.86084160183337e-05\\
25	-5.98976144521513e-06\\
26	1.36394653703376e-05\\
27	2.62375439937724e-05\\
28	3.50949112404578e-05\\
29	4.21043650139466e-05\\
};
\addplot[const plot, color=istblue, forget plot] table[row sep=crcr] {%
0	1\\
1	-1.06450619511143\\
2	-0.544889438302398\\
3	-0.766195255978025\\
4	-0.363837015408706\\
5	-0.162228273839408\\
6	-0.151563650310973\\
7	-0.0657380501096553\\
8	-0.0246852857401762\\
9	-0.00439225943661161\\
10	0.00600717221697318\\
11	0.0115466835270904\\
12	-0.00541762158707693\\
13	0.00412338162179456\\
14	0.00830745249904527\\
15	0.00118309835298791\\
16	0.00206908797190053\\
17	0.00154800801456914\\
18	0.00089657519966584\\
19	0.000359861282267438\\
20	-2.66344574785958e-05\\
21	0.00210379129347649\\
22	0.00100554112038933\\
23	0.000523972166152171\\
24	0.000309676451615011\\
25	0.00021257355360274\\
26	0.000167643149345301\\
27	0.00014638594873063\\
28	0.000136115251404676\\
29	0.00013106834166828\\
};
\addplot[const plot, color=istred, forget plot] table[row sep=crcr] {%
0	0\\
1	-0.0279120305555201\\
2	-0.137329304618717\\
3	-0.222317841264959\\
4	-0.141375929101346\\
5	-0.0746780078753852\\
6	-0.0380589089524459\\
7	-0.0187249879050653\\
8	-0.00552000365039691\\
9	0.00303627233568767\\
10	0.00837001091459055\\
11	0.0115614754902135\\
12	0.0119762007892326\\
13	0.0112948438438771\\
14	0.0110726059217441\\
15	0.010353026204529\\
16	0.00605593915879288\\
17	0.00333439460138775\\
18	0.00175295464975895\\
19	0.00083292801712631\\
20	0.00030290744845503\\
21	0.000170910842312688\\
22	0.000202936661242743\\
23	0.000196148813601954\\
24	0.00018163578279577\\
25	0.000169137894378578\\
26	0.000160759821407603\\
27	0.000156207630472959\\
28	0.000154673085440998\\
29	0.000155435475446077\\
};
\end{axis}
\end{tikzpicture}%
	\caption{State of the first (top) and second (bottom) batch reactor.}
	\label{fig_act_state}
\end{figure}
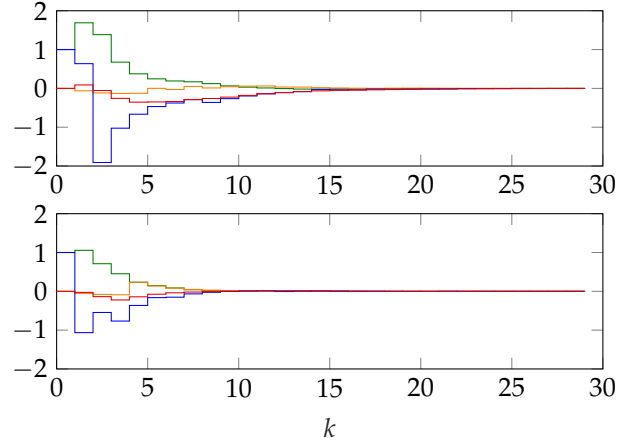
\begin{figure}
	\centering
	% This file was created by matlab2tikz.
%
%The latest updates can be retrieved from
%  http://www.mathworks.com/matlabcentral/fileexchange/22022-matlab2tikz-matlab2tikz
%where you can also make suggestions and rate matlab2tikz.
%
\definecolor{mycolor1}{rgb}{0.00000,0.44700,0.74100}%
\begin{tikzpicture}

\begin{axis}[%
width=\columnwidth,
height=3cm,
xmin=0,
xmax=30,
xlabel style={font=\color{white!15!black}},
ymin=0,
ymax=0.3,
ylabel style={font=\color{white!15!black}},
axis background/.style={fill=white}
]
\addplot[const plot, color=istgreen, forget plot] table[row sep=crcr] {%
0	0\\
1	0\\
2	0\\
3	0\\
4	0.253740607581983\\
5	0\\
6	0.205277512716743\\
7	0\\
8	0.115879965770164\\
9	0.0906324268959922\\
10	0.0760857932629774\\
11	0\\
12	0.0340187239450731\\
13	0\\
14	0\\
15	0\\
16	0\\
17	0.0124621988315757\\
18	0.0101304168873865\\
19	0.00789183316027064\\
20	0\\
21	0\\
22	0.00320781998779307\\
23	0.00213293325681041\\
24	0\\
25	0.00114760218020349\\
26	0.000916565352934939\\
27	0.000731668428797572\\
28	0\\
29	0\\
};
\end{axis}

\begin{axis}[%
width=\columnwidth,
height=3cm,
at={(0cm,-2cm)},
xmin=0,
xmax=30,
xlabel style={font=\color{white!15!black}},
ymin=0,
ymax=15,
ylabel style={font=\color{white!15!black}},
axis background/.style={fill=white}
]
\addplot[const plot, color=istgreen, forget plot] table[row sep=crcr] {%
0	0\\
1	10.2359577043621\\
2	0\\
3	0\\
4	0\\
5	0\\
6	0\\
7	0.505350945959599\\
8	0\\
9	0\\
10	0\\
11	0\\
12	0\\
13	-0.154289999855366\\
14	0\\
15	0\\
16	-0.0835338414528943\\
17	0\\
18	0\\
19	0\\
20	0\\
21	-0.0410865301296355\\
22	0\\
23	0\\
24	-0.0141546323704795\\
25	0\\
26	0\\
27	0\\
28	-0.00339452085295799\\
29	0\\
};
\end{axis}

\begin{axis}[%
width=\columnwidth,
height=3cm,
at={(0cm,-4cm)},
xmin=0,
xmax=30,
xlabel style={font=\color{white!15!black}},
ymin=0,
ymax=1,
ylabel style={font=\color{white!15!black}},
axis background/.style={fill=white}
]
\addplot[const plot, color=istorange, forget plot] table[row sep=crcr] {%
0	0\\
1	0\\
2	0\\
3	0.685891963108391\\
4	0\\
5	0\\
6	0\\
7	0\\
8	0\\
9	0\\
10	0\\
11	0\\
12	0\\
13	0\\
14	0\\
15	-0.0163973798319061\\
16	0\\
17	0\\
18	0\\
19	0\\
20	0\\
21	0\\
22	0\\
23	0\\
24	0\\
25	0\\
26	0\\
27	0\\
28	0\\
29	0\\
};
\end{axis}

\begin{axis}[%
width=\columnwidth,
height=3cm,
at={(0cm,-6cm)},
xmin=0,
xmax=30,
xlabel style={font=\color{white!15!black}},
xlabel={$k$},
ymin=-0.0101696960947368,
ymax=10,
ylabel style={font=\color{white!15!black}},
axis background/.style={fill=white}
]
\addplot[const plot, color=istorange, forget plot] table[row sep=crcr] {%
0	7.23143867120409\\
1	0\\
2	1.851945220797\\
3	0\\
4	0\\
5	0.379438859859105\\
6	0\\
7	0\\
8	0\\
9	0\\
10	0\\
11	0.0851871591948722\\
12	0\\
13	0\\
14	0.0383064757746671\\
15	0\\
16	0\\
17	0\\
18	0\\
19	0\\
20	-0.0101696960947368\\
21	0\\
22	0\\
23	0\\
24	0\\
25	0\\
26	0\\
27	0\\
28	0\\
29	-0.0010858375186098\\
};
\end{axis}
\end{tikzpicture}%
	\caption{The two inputs of the first (green) and second (orange) batch reactor.}
	\label{fig_act_input}
\end{figure}
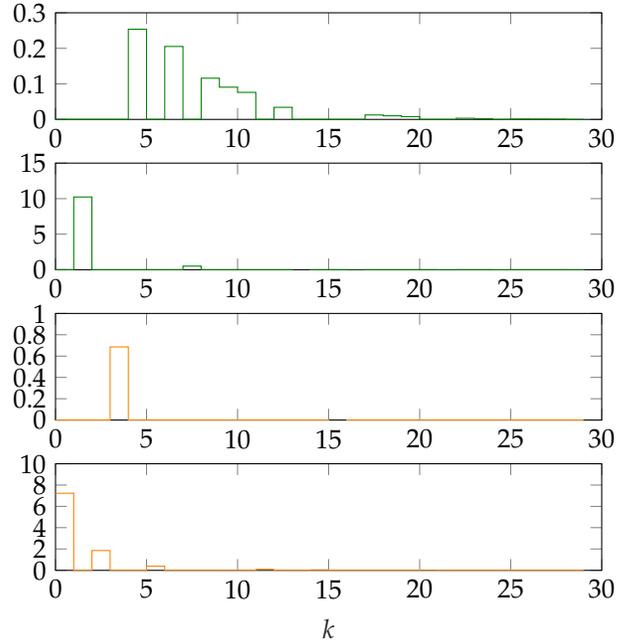
%\begin{table}
%	\centering
%	\caption{Total number of times a respective input was triggered}
%	\begin{tabular}{c c c c c}
%		Input number & 1 & 2 & 3 & 4 \\ \hline
%		Number of times triggered & 14 & 7 & 2 & 7 
%	\end{tabular}
%\end{table}

The terminal costs and controllers were computed by solving the LMIs given in Theorem \ref{thm_convergence_act}. Simulation results for this setup with a prediction horizon of $N=3$ can be found in Figures \ref{fig_act_state} and \ref{fig_act_input}. It can be seen that the states of the second batch reactor, which have a higher weight, are kept lower and converge faster than those of the first reactor. 
The first actuator of the first reactor, which features the highest cost, is chosen for transmission many more times than the other inputs. This is due to the fact that the high weight and the set-to-zero strategy make it favorable to send low control values frequently. In contrast, for the lowly weighted second input, there are fewer triggerings which come with a rather high control input.  For the equally weighted inputs of the second reactor, the behavior is reversed: the first actuator of the second reactor is triggered much less frequently than the second one and the input levels are higher, respectively lower than those of the first reactor. In summary, the closed loop trajectories resulting from an application of the MPC scheme reflect the behavior expected from the chosen weightings. Despite the fact that one of the actuators is given most of the overall network access, by virtue of the terminal costs used in the MPC, convergence of the state of both plants under the proposed scheduling and control scheme to zero is guaranteed.

\section{Summary and Outlook}
\label{sec_summary}

In this paper, we proposed the novel approach to use an MPC with time-varying terminal ingredients to co-design the transmission schedule and control inputs in rollout approaches. Compared to the formerly employed multi-step MPC formulation, the intended benefit was to lower computational complexity and provide better robustness. The formulation and computation of the required time-varying terminal ingredients is, however, slightly more involved. For two distinct NCS setups, an explicit choice of the terminal regions and costs was provided to ensure convergence of the setup. Eventually, a numerical example demonstrated that indeed, the computational complexity could be drastically lowered in comparison to multi-step MPC formulations.

Possible future work could extend the presented approaches to construct terminal regions and costs to nonlinear plants, as discussed in Remark \ref{remark_nonlinear}. Further, guaranteed convergence of rollout approaches in the case of incomplete state information would be a highly interesting future topic.

\small
\bibliographystyle{plain}
\bibliography{bib_ACC20}

\end{document}